\documentclass[sigconf]{acmart}

\copyrightyear{2026}
\acmYear{2026}
\setcopyright{cc}
\setcctype{by}
\acmConference[KDD '26]{Proceedings of the 32nd ACM SIGKDD Conference on Knowledge Discovery and Data Mining V.2}{August 09--13, 2026}{Jeju Island, Republic of Korea}
\acmBooktitle{Proceedings of the 32nd ACM SIGKDD Conference on Knowledge Discovery and Data Mining V.2 (KDD '26), August 09--13, 2026, Jeju Island, Republic of Korea}
\acmDOI{10.1145/3770855.3818832}
\acmISBN{979-8-4007-2259-2/2026/08}

\usepackage{multirow}
\usepackage{diagbox}
\usepackage{tikz}
\usepackage{graphicx}
\usepackage{caption}
\usepackage{subcaption}

\usepackage{natbib}

\usepackage[linesnumbered,vlined,ruled,commentsnumbered]{algorithm2e}
\newlength\mylen

\makeatletter
\newcommand{\algorithmfootnote}[2][\footnotesize]{%
  \let\old@algocf@finish\@algocf@finish
  \def\@algocf@finish{\old@algocf@finish
    \leavevmode\rlap{\begin{minipage}{\linewidth}
    #1#2
    \end{minipage}}%
  }%
}

\usepackage{makecell}
\usepackage{float}
\usepackage{placeins}
\usepackage[nameinlink]{cleveref}
\Crefname{figure}{Figure}{Figures}
\crefname{figure}{Figure}{Figures}
\crefname{example}{Example}{Example}
\crefname{theorem}{Theorem}{Theorem}
\crefname{corollary}{Corollary}{Corollary}
\crefname{lemma}{Lemma}{Lemma}
\crefname{proposition}{Proposition}{Proposition}
\crefname{assumption}{Assumption}{Assumption}
\crefname{section}{Section}{Section}
\crefname{algorithm}{Algorithm}{Algorithm}

\newcommand{\scaleeq}[2][0.95]{\scalebox{#1}{$#2$}}

\usepackage{enumitem}
\newlist{propenum}{enumerate}{1} 
\setlist[propenum]{label=\alph*{\rm)}, ref=\theproposition(\alph*)}
\crefalias{propenumi}{proposition}
\newlist{corenum}{enumerate}{1} 
\setlist[corenum]{label=\alph*{\rm)}, ref=\thecorollary(\alph*)}
\crefalias{corenumi}{corollary}
\newlist{lemenum}{enumerate}{1} 
\setlist[lemenum]{label=\alph*{\rm)}, ref=\thelemma(\alph*)}
\crefalias{lemenumi}{lemma}

\usepackage{amsthm,thmtools}
\declaretheorem[name=Theorem,style=definition,numberwithin=section]{theorem}
\declaretheorem[name=Definition,style=definition]{definition}

\declaretheorem[name=Proposition,numberlike=theorem]{proposition}

\declaretheorem[name=Lemma,numberlike=theorem]{lemma}
\declaretheorem[name=Remark,style=definition,numberwithin=section]{remark}

\usepackage{amsmath,amsthm,amsfonts,bm,amssymb,mathtools}

\newcommand{\bsmat}{\left[\begin{smallmatrix}}
\newcommand{\esmat}{\end{smallmatrix}\right]}

\DeclareMathAlphabet{\mathsfit}{\encodingdefault}{\sfdefault}{m}{sl}
\SetMathAlphabet{\mathsfit}{bold}{\encodingdefault}{\sfdefault}{bx}{n}

\usepackage{svg}
\usepackage{microtype}
\usepackage{booktabs}

\usepackage{hyperref}
\usepackage{url}
\usepackage{bbm}
\usepackage{tcolorbox}

\usepackage{wrapfig}

\usepackage[version=4]{mhchem}

\begin{document}

\title{Beyond Pairwise Interactions: Equivariant Hypergraph Diffusion for Crystal Structure Prediction}





\author{Yang Liu}
\orcid{0000-0003-3791-4343}
\affiliation{%
  \institution{Academy of Mathematics and Systems Science, Chinese Academy of Sciences}
  \department{}
  \streetaddress{}
  \city{Beijing}
  \state{}
  \postcode{}
  \country{China}
}
\email{liuyang2020@amss.ac.cn}

\author{Chuan Zhou}
\authornote{Corresponding authors.}
\orcid{0000-0001-9958-8673}
\affiliation{%
  \institution{Academy of Mathematics and Systems Science, Chinese Academy of Sciences}
  \department{}
  \streetaddress{}
  \city{Beijing}
  \state{}
  \postcode{}
  \country{China}}
\affiliation{%
  \institution{School of Cyber Security, University of Chinese Academy of Sciences}
  \department{}
  \streetaddress{}
  \city{Beijing}
  \state{}
  \postcode{}
  \country{China}}
\email{zhouchuan@amss.ac.cn}

\author{Shuai Zhang}
\orcid{0009-0001-3152-3651}
\affiliation{%
  \institution{Academy of Mathematics and Systems Science, Chinese Academy of Sciences}
  \department{}
  \streetaddress{}
  \city{Beijing}
  \state{}
  \postcode{}
  \country{China}}
\email{zhangshuai2021@amss.ac.cn}

\author{Xiaotong Wu}
\orcid{0009-0007-3857-0390}
\affiliation{%
  \institution{College of Materials Science and Engineering, Donghua University}
  \department{}
  \streetaddress{}
  \city{Shanghai}
  \state{}
  \postcode{}
  \country{China}}
\email{2210650@mail.dhu.edu.cn}

\author{Peng Zhang}
\orcid{0000-0001-7973-2746}
\affiliation{%
  \institution{Cyberspace Institute of Advanced Technology, Guangzhou University}
  \department{}
  \streetaddress{}
  \city{Guangzhou}
  \state{}
  \postcode{}
  \country{China}}
\email{p.zhang@gzhu.edu.cn}

\author{Xixun Lin}
\orcid{0009-0004-6645-0597}
\affiliation{%
  \institution{Institute of Information Engineering, Chinese Academy of Sciences}
  \department{}
  \streetaddress{}
  \city{Beijing}
  \state{}
  \postcode{}
  \country{China}
}
\email{linxixun@iie.ac.cn}

\author{Shirui Pan}
\orcid{0000-0003-0794-527X}
\affiliation{%
  \institution{Griffith University}
  \department{}
  \streetaddress{}
  \city{Gold Coast}
  \state{QLD}
  \postcode{4215}
  \country{Australia}
}
\email{s.pan@griffith.edu.au}

\author{Zhao Li}
\orcid{0000-0002-5056-0351}
\authornotemark[1]
\affiliation{%
  \institution{Zhejiang Lab}
  \department{}
  \streetaddress{}
  \city{Hangzhou}
  \state{}
  \postcode{}
  \country{China}
}
\affiliation{%
  \institution{R\&D, Hangzhou Yugu Technology Co., Ltd.}
  \department{}
  \streetaddress{}
  \city{Hangzhou}
  \state{}
  \postcode{}
  \country{China}
}
\email{lzjoey@gmail.com}

\renewcommand{\shortauthors}{Yang Liu et al.}

\begin{abstract} 
Crystal Structure Prediction (CSP) remains a fundamental challenge with significant implications for materials discovery and the advancement of various scientific disciplines. Recent advances have demonstrated that generative models, particularly diffusion models, are especially promising for CSP. However, traditional graph-based representations, where atomic bonds are modeled as pairwise graph edges, fail to capture the intricate high-order interactions essential for accurately describing crystal structures. To address this limitation, we propose leveraging hypergraphs to represent crystal structures, enabling more expressive modeling of multi-way atomic interactions. Hypergraphs naturally encode complex high-order relationships and respect key symmetries---such as permutation and periodic translation invariance---that are crucial for characterizing crystalline materials. Building on this representation, we propose the \textbf{E}quivariant \textbf{H}ypergraph \textbf{Diff}usion Model (\textbf{EH-Diff}), a generative framework designed to exploit the symmetry-preserving properties of hypergraphs. EH-Diff provides an efficient and accurate method for predicting crystal structures, with rigorous theoretical guarantees on invariance preservation. Empirically, we conduct extensive experiments on four benchmark datasets, and the results demonstrate that EH-Diff outperforms state-of-the-art CSP methods even with a single diffusion sample.
\end{abstract}






\maketitle

\section{Introduction}\label{sec:introduction}

Crystal Structure Prediction (CSP), a long-standing challenge in the physical sciences since the early 2000s \cite{desiraju2002cryptic, woodley2008crystal}, aims to determine the stable 3D structure of a compound based solely on its chemical composition. Given the fundamental role of crystals in a wide range of materials \cite{steiger2005crystal, cartwright2012beyond}, accurately predicting their spatial configuration is crucial \cite{wooster1953physical, sato2016dynamic} for understanding the physical and chemical properties that govern their applications in various academic and industrial domains, including drug development \cite{zhang2023liquid}, battery technologies \cite{xie2018crystal, yin2024utilizing}, and catalysis \cite{bao2023microwave}.

\begin{figure}[b]
    \centering
    \includegraphics[width=.5\linewidth]{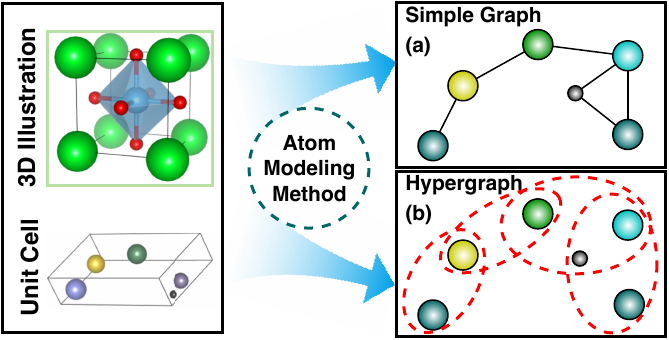}
    \caption{Illustration of Crystal Structure Modeling. (a) Simple Graph and (b) Hypergraph.
    Interactions within the hypergraph demonstrate effects involving multiple atoms, with pairwise (first-order) interactions evident in (a), while higher-order (non-pairwise) interactions are captured in (b).}
    \label{fig:EH-Diff-illustration}
\end{figure}

Conventional CSP methods typically employ density functional theory (DFT) \cite{kohn1965self} to iteratively evaluate the energy of the system. Techniques like random search \cite{pickard2011ab} and Bayesian optimization \cite{yamashita2018crystal} explore the energy landscape to identify stable local or global minima \cite{oganov2019structure}. 
Despite these efforts, prediction of crystal properties remains challenging because of the inherent complexity of crystalline materials. Crystalline structures, characterized by periodic repetition of a unit cell \cite{yan2022periodic} in a 3D lattice, feature highly ordered atomic arrangements. The physical properties of these materials, including electronic, mechanical, and thermal characteristics, are intrinsically linked to their atomic configurations.

\begin{figure*}[t]
    \centering
    \includegraphics[width=0.95\linewidth]{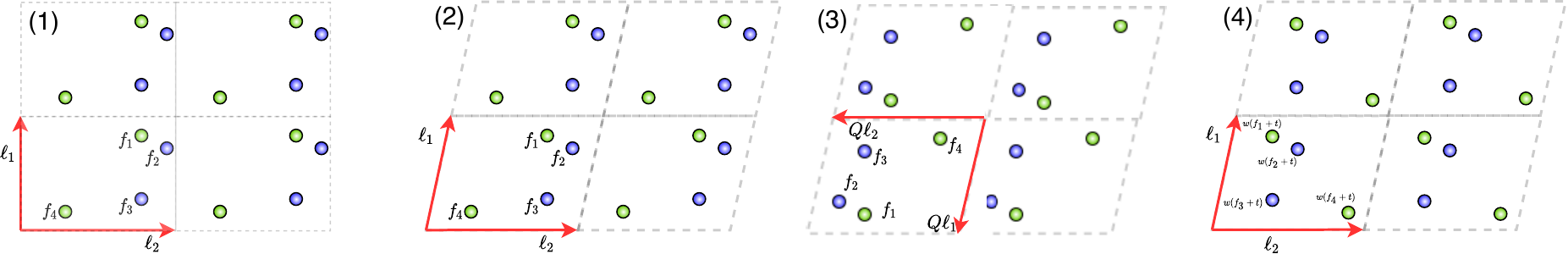}
    \caption{{Crystal Structure Illustration.} (1) crystal structure in a Cartesian coordinate system; (2) crystal structure in a fractional coordinate system. (2)$\to$(3) orthogonal transformation of the lattice vectors; (2)$\to$(4) periodic translation of the fractional coordinates.}
    \label{fig:graph}
\end{figure*}

Recent advances in AI-driven materials research have leveraged graph-based deep learning models \citep{xie2018crystal, xie2021crystal}, with a particular focus on diffusion models \citep{lin2024equicsp, jiao2024crystal}. These approaches employ graph-based diffusion techniques to generate atomic configurations while preserving structural invariants. However, traditional graph-based approaches, including graph neural networks and graph transformers, are inherently limited to modeling pairwise relationships \citep{kim2022equivariant}, where atoms are represented as vertices and atomic bonds as edges. This constraint prevents them from capturing higher-order interactions that are crucial for representing the structural complexity of crystal materials. Crystal structures often exhibit multi-body effects, such as shared lattice points and multi-atomic coordination \citep{voet1970crystal, grimmer1974coincidence}, which cannot be comprehensively represented using only pairwise edges. For instance, in perovskite structures (\ce{ABX3}), the B-site cation is coordinated by six X anions forming an octahedron---a quintessential multi-body interaction that pairwise edges alone cannot fully capture.

To overcome this limitation, we employ \emph{hypergraphs}, a generalized graph structure capable of capturing high-order relationships, to model crystal materials. Hypergraphs offer a key advantage over traditional graph representations by naturally encoding multi-atom interactions through hyperedges that connect more than two atoms simultaneously. This enables a more expressive representation of the intricate coordination environments and structural dependencies within crystal structures, facilitating the analysis of crystal stability, coordination geometry, and bond distributions. Moreover, the ability of hypergraphs to inherently incorporate symmetries and periodicity makes them particularly well-suited for modeling the fundamental properties of crystalline materials.

However, integrating hypergraphs into generative models for CSP poses several \emph{challenges}. First, constructing effective hypergraph-based material representations requires principled strategies for defining hyperedges that capture physically meaningful multi-atom interactions. Second, developing diffusion models that operate on hypergraph-structured data is non-trivial: the model must capture complex multi-atomic interactions while preserving key structural invariants---such as permutation and periodic translation symmetries---inherent in crystalline materials. Furthermore, these symmetries must be maintained throughout both the forward and reverse diffusion processes, necessitating equivariant architectures capable of handling the high-order interactions encoded in hypergraphs.

In this work, we address the CSP problem through a hypergraph-based modeling approach and propose the Equivariant Hypergraph Diffusion Model (EH-Diff), a novel framework that captures high-order interactions while preserving the intrinsic symmetries of crystal structures (defined in \Cref{section:Symmetries}). The model simultaneously generates the lattice vectors and fractional coordinates of all atoms via a denoising hypergraph architecture with periodic translation equivariance (detailed in \Cref{section:EH-Diff}). Through iterative refinement, our framework captures complex structural dependencies while provably maintaining invariance properties (theoretical justification in \Cref{section:Denoising}). Comprehensive experiments across four benchmark datasets demonstrate that EH-Diff consistently outperforms existing graph-based diffusion models for CSP.

Our contributions are summarized as follows:

\begin{itemize}[leftmargin=*,nosep]
    \item We formalize the CSP problem within a hypergraph modeling framework that describes high-order atomic interactions. To the best of our knowledge, this is the first work to integrate hypergraph representations with generative diffusion models for CSP, offering a novel perspective on capturing complex multi-atom coordination environments.
    \item We propose the \textbf{E}quivariant \textbf{H}ypergraph \textbf{Diff}usion Model ({\bf EH-Diff}), a framework that unifies hypergraph-based representation learning with a diffusion process. EH-Diff captures high-order atomic interactions while provably preserving key symmetries, including permutation and periodic translation invariance, which are fundamental properties of crystal structures.
    \item We evaluate the proposed framework through extensive experiments on four benchmark datasets, demonstrating its superior performance over state-of-the-art generative models even with a single diffusion sample.
\end{itemize}

\section{Problem Setup and Analysis}

In this section, we formally define the problem and provide a concise theoretical analysis. Detailed mathematical preliminaries are provided in \Cref{appendix:Preliminaries}.

\subsection{Representation of Crystal Structure}\label{section:Representation}

We introduce two frameworks for representing crystal structures and define the CSP task, providing a foundation for modeling atomic configurations in crystal lattices and predicting material behaviors.

\paragraph{Cartesian Coordinate System} A three-dimensional crystal is defined as an infinite periodic arrangement of atoms in space, characterized by the smallest repeating unit called the unit cell. The unit cell can be represented by a triplet $\mathbf{M}_C = (\mathbf{A}, \mathbf{X}, \mathbf{L})$, where $\mathbf{A} = [\mathbf{a}_1, \mathbf{a}_2, \ldots, \mathbf{a}_N] \in \mathbb{R}^{h \times N}$ denotes the set of one-hot encoded representations of atom types, $\mathbf{X} = [\mathbf{x}_1, \mathbf{x}_2, \ldots, \mathbf{x}_N] \in \mathbb{R}^{3 \times N}$ represents the Cartesian coordinates of the atoms, and $\mathbf{L} = [\ell_1, \ell_2, \ell_3] \in \mathbb{R}^{3 \times 3}$ is the lattice matrix, which contains three basis vectors that describe the periodicity of the crystal. The infinite periodic structure of the crystal can be formally expressed as:
\begin{align}
\{ (\mathbf{a}'_i, \mathbf{x}'_i) | \mathbf{a}'_i = \mathbf{a}_i, \, \mathbf{x}'_i = \mathbf{x}_i + \mathbf{L}\mathbf{k}, \, \forall \mathbf{k} \in \mathbb{Z}^{3} \},
\end{align}
where the $j$-th element of the integer vector $\mathbf{k}$ represents the integer translation along the $j$-th basis vector $\ell_j$ in three-dimensional space for $j\in\{1,2,3\}$.

\paragraph{Fractional Coordinate System} Analogous to the Cartesian coordinate system, a unit cell can also be defined by a triplet $\mathbf{M}_F = (\mathbf{A}, \mathbf{F}, \mathbf{L})$ in the fractional coordinate system. In this system, a point represented by the fractional coordinate vector $\mathbf{f} = [f_1, f_2, f_3]^\top \in [0, 1)^3$ corresponds to the Cartesian coordinate vector $\mathbf{x}$, given by $\mathbf{x} = \sum_{i=1}^{3} f_i \ell_i$, where $\ell_i$ denotes the $i$-th basis vector of the lattice.

\begin{remark}
We adopt the fractional coordinate system, representing the crystal as $\mathbf{M} = \mathbf{M}_F = (\mathbf{A}, \mathbf{F}, \mathbf{L})$, where the fractional coordinates of atoms in the unit cell form the matrix $\mathbf{F} \in [0, 1)^{3 \times N}$. The two frameworks are shown in \Cref{fig:graph}.
\end{remark}

\begin{remark}
The Cartesian coordinate system $\mathbf{X}$ uses three mutually orthogonal basis vectors as coordinate axes. In crystallography, the fractional coordinate system is commonly employed to capture the periodicity of crystal structures \cite{nouira2018crystalgan, ren2022invertible}, using lattice vectors $(\ell_1, \ell_2, \ell_3)$ as a basis.      
\end{remark}

\paragraph{Crystal Structure Prediction (CSP)} CSP aims to predict the lattice matrix $\mathbf{L}$ and the fractional coordinate matrix $\mathbf{F}$ from the chemical composition $\mathbf{A}$. The goal is to learn the joint conditional distribution $p(\mathbf{L}, \mathbf{F} | \mathbf{A})$, where $\mathbf{L} \in \mathbb{R}^{3 \times 3}$ is the lattice matrix that defines the unit cell's periodic structure, and $\mathbf{F} \in [0,1)^{3 \times N}$ is the fractional coordinate matrix whose $i$-th column $\mathbf{f}_i$ represents the fractional coordinates of the $i$-th atom. Formally, CSP involves learning $p(\mathbf{L}, \mathbf{F} | \mathbf{A})$ such that the generated lattice and atomic configurations faithfully reproduce the ground-truth crystal structure.

\subsection{Hypergraph Modeling of Crystal Structure}

\begin{definition}[Hypergraph]
Consider a set of vertices \(\mathbf{V} = \{\mathbf{v}_i\}_{i=1}^n\) connected by hyperedges \(\mathbf{E} = \{\mathbf{e}_k\}_{k=1}^m\), forming a hypergraph \(\mathbf{H} = (\mathbf{V}, \mathbf{E})\) consisting of \(n\) vertices and \(m\) hyperedges, where each hyperedge can connect any number of vertices. 
\end{definition}
Let \(\mathbf{X} = [\mathbf{x}_1, \ldots, \mathbf{x}_n]^T \in \mathbb{R}^{n \times d}\) represent the vertex attributes, with \(\mathbf{x}_v \in \mathbb{R}^d\) indicating the attribute of vertex \(v\). The degree of vertex \(v\), denoted as \(d_v\), is defined by \(d_v = |\{\mathbf{e} \in \mathbf{E} : \mathbf{v} \in \mathbf{e}\}|\). Moreover, let \(\mathbf{D}\) and \(\mathbf{D}_e\) represent the diagonal degree matrix for vertices in \(\mathbf{V}\) and the sub-matrix corresponding to vertices within a hyperedge \(\mathbf{e}\), respectively. A more formal definition of hypergraph can be found in \Cref{appendix:Preliminaries}.

Building on the previous definition of hypergraph \citep{kim2022equivariant}, we extend this concept and its key attributes to the field of crystallography. This extension provides a more detailed framework for understanding the structural and dynamic properties of crystalline materials, offering new analytical tools and methodologies that enhance the study of crystal structures.

\begin{definition}[Hypergraph-based Crystal]
Formally, let the hypergraph $\mathbf{H}_{\text{crystal}} = (\mathbf{M}, \mathbf{V}, \mathbf{E})$ represent a crystal unit cell, where $\mathbf{M}$ denotes the unit cell, $\mathbf{V}$ represents the set of atoms (vertices), and $\mathbf{E}$ denotes the set of hyperedges capturing the bonds or interactions between these atoms, potentially involving more than two atoms simultaneously.
\end{definition}

\begin{remark}
The hypergraph-based framework represents crystal lattices by mapping atoms to vertices and multi-atom interactions to hyperedges, enabling the modeling of complex relationships beyond pairwise bonding. This representation facilitates the analysis of topological properties, collective atomic behaviors, and coordination geometries within crystalline materials.
\end{remark}

Traditional pairwise representations fundamentally miss critical multi-body effects in crystals:

\noindent(1) \textbf{Coordination Polyhedra}: In spinel structures (\ce{AB2O4}), the tetrahedral (\ce{AO4}) and octahedral (\ce{BO6}) coordination cannot be decomposed into independent pairwise bonds---the collective geometry determines stability.\\
(2) \textbf{Ring Structures}: In zeolites, the stability depends on complete \ce{Si-O-Si} ring closures (4-, 6-, or 8-membered rings), not individual \ce{Si-O} bonds.\\
(3) \textbf{Jahn--Teller Distortions}: The cooperative distortion in perovskites involves correlated movements of multiple atoms that pairwise models cannot capture.

Hyperedges naturally encode these multi-body correlations: a single hyperedge $e = \{\text{B}, \text{O}_1, \text{O}_2, \text{O}_3, \text{O}_4, \text{O}_5, \text{O}_6\}$ represents the complete octahedral coordination, enabling the model to learn collective geometric constraints.

\subsection{Symmetries of Crystal Structure}\label{section:Symmetries}

Crystals differ fundamentally from molecules in that they consist of a repeating unit cell organized within a regular three-dimensional lattice. This periodic structure endows crystals with unique symmetries, including invariance under the orthogonal group $O(3)$ and periodic translational invariance within the fractional coordinate system. These symmetries play a critical role in defining the physical and quantum mechanical properties of crystalline materials, distinguishing them from discrete molecular entities.

\begin{definition}[Atom Permutation Invariance]
For any permutation $\mathbf{P} \in \mathbb{S}_{N}$ (where $\mathbb{S}_{N}$ represents the symmetric group of degree $N$), the distribution \(p(\mathbf{L}, \mathbf{F} | \mathbf{A})\) remains invariant under the transformation \(p(\mathbf{L}, \mathbf{F}\mathbf{P} | \mathbf{A}\mathbf{P})\). This indicates that reordering the atoms does not affect the underlying distribution, implying the inherent permutation symmetry of the atomic configuration.
\end{definition}

\begin{definition}[$O(3)$ Invariance]
    For any orthogonal transformation \(\mathbf{Q} \in \mathbb{R}^{3 \times 3}\) satisfying \(\mathbf{Q}^{\top} \mathbf{Q} = \mathbf{I}\), the distribution \(p(\mathbf{Q} \mathbf{L}, \mathbf{F} | \mathbf{A})\) is equivalent to \(p(\mathbf{L}, \mathbf{F} | \mathbf{A})\). This invariance indicates that any rotation or reflection of the lattice \(\mathbf{L}\) does not alter the distribution.
\end{definition}

\begin{definition}[Periodic Translation Invariance]
    For any translation \(\mathbf{T} \in \mathbb{R}^{3}\), the distribution \(p(\mathbf{L}, w(\mathbf{F} + \mathbf{T} \mathbf{1}^{\top}) | \mathbf{A})\) is equivalent to \(p(\mathbf{L}, \mathbf{F} | \mathbf{A})\). Here, the function \(w(\mathbf{F}) = \mathbf{F} - \lfloor \mathbf{F} \rfloor \in [0, 1)^{3 \times N}\) returns the fractional part of each element in \(\mathbf{F}\), and \(\mathbf{1} \in \mathbb{R}^{3}\) is a vector with all elements set to one. This invariance indicates that any periodic translation of \(\mathbf{F}\) does not alter the distribution \(p(\mathbf{L}, \mathbf{F} | \mathbf{A})\).
\end{definition}

\begin{proposition}
    The hypergraph-based crystal representation $\mathbf{H}_{\text{crystal}} = (\mathbf{M}, \mathbf{V}, \mathbf{E})$ is permutation invariant: for any permutation $\mathbf{P} \in \mathbb{S}_{N}$ acting on the vertex set $\mathbf{V}$, the total potential energy $U(\mathbf{H}_{\text{crystal}})$ satisfies $U(\mathbf{P} \cdot \mathbf{H}_{\text{crystal}}) = U(\mathbf{H}_{\text{crystal}})$, since the energy contributions from hyperedges depend only on the unordered set of participating atoms, not on their labeling.
\end{proposition}

\begin{remark}
The crystal symmetries defined above are consistent with those in the Cartesian coordinate system \cite{xie2021crystal, yan2022periodic}. Specifically, rotations or reflections of the lattice $\mathbf{L}$ do not affect the fractional coordinates of atoms. Consequently, the full $E(3)$ invariance in Cartesian coordinates reduces to $O(3)$ invariance of $\mathbf{L}$ alone in fractional coordinates. Moreover, since the fractional coordinate system inherently encodes the periodicity of the crystal structure, periodic translation invariance admits a more direct formulation without introducing corner atoms, as previously required under Cartesian coordinates.
\end{remark}

\subsection{Why Hypergraphs? Comparison with Alternative Representations}
\label{sec:why_hypergraph}
Before proceeding with our hypergraph-based approach, we briefly contrast it with alternative high-order representations to justify our modeling choice. 

Subgraph-based methods enumerate recurring local patterns such as cliques to capture multi-body interactions, but face fundamental scalability challenges as the enumeration space grows combinatorially with coordination shell size. Motif-based approaches identify chemically meaningful structural units and work effectively for molecular systems with well-defined functional groups, yet struggle with crystalline materials where coordination geometries vary continuously and structural motifs frequently span periodic unit cell boundaries. Persistent homology methods offer powerful tools for capturing multi-scale topological features, but they aggregate local geometric information into global invariants, thereby discarding the atom-level positional details that are essential for generative structure prediction.

Hypergraphs address these limitations through a fundamentally different representation strategy. Rather than decomposing multi-body interactions into collections of lower-order terms or abstracting them into topological summaries, hypergraphs directly encode coordination environments as unified mathematical objects. This native representation of multi-atom relationships preserves the geometric and chemical integrity of local structures while maintaining computational tractability through sparse connectivity. Moreover, the mathematical framework of hypergraphs naturally accommodates the symmetry requirements of crystalline materials, allowing equivariant message passing to respect both permutation invariance and periodic translation symmetry within a single coherent formalism.

\section{The Proposed Model: EH-Diff}\label{section:EH-Diff}

\subsection{Diffusion Model on Crystal Structure}

{\bf Diffusion on $\mathbf{F}$.} To respect the bounded periodic domain of $\mathbf{F}$, we employ the score-based generative model \cite{song2020improved} with the wrapped normal distribution\footnote{The wrapped normal distribution extends the standard normal to periodic domains by summing over all integer translations, ensuring that $\mathbf{f}$ and $\mathbf{f} + \mathbf{k}$ ($\mathbf{k} \in \mathbb{Z}^3$) are treated as equivalent positions.} \cite{de2022riemannian}, which has proven effective for molecular conformer generation \cite{jing2022torsional}. The forward diffusion process is defined as $q(\mathbf{F}_{t} | \mathbf{F}_{0}) \propto \sum_{\mathbf{Z} \in \mathbb{Z}^{3} \times N} \exp\left(-\frac{\|\mathbf{F}_{t} - \mathbf{F}_{0} + \mathbf{Z}\|_{F}^{2}}{2 \sigma_{t}^{2}}\right)$. Here, $\sigma_{t}$ is an exponential annealing schedule, and $\mathbf{F}_{t} = w(\mathbf{F}_{0} + \sigma_{t} \epsilon_{\mathbf{F}})$ with $\epsilon_{\mathbf{F}} \sim \mathcal{N}(\mathbf{0}, \mathbf{I})$. $q(\mathbf{F}_{t} | \mathbf{F}_{0})$ approaches the uniform distribution $\mathcal{U}(\mathbf{0}, 1)$ when $T$ is sufficiently large. For the backward process, we apply the predictor-corrector method \citep{song2020improved} to reconstruct $\mathbf{F}_{0}$ from $\mathbf{F}_{T} \sim \mathcal{U}(\mathbf{0}, 1)$. The estimated score $\hat{\epsilon}_{\mathbf{F}}$ is obtained from the denoising model $\phi_{\theta}(\mathbf{A}, \mathbf{F}_{t}, \mathbf{L}_{t}, t)$. The training loss for score matching is given as:
\begin{align}\label{equation:lossF}
\scaleeq[0.96]{\mathcal{L}_{\mathbf{F}} = \mathbb{E}_{\mathbf{F}_{t} \sim q(\mathbf{F}_{t} | \mathbf{F}_{0}), t \sim \mathcal{U}(1, T)}
\Big{[} \lambda_{t} \left\| \nabla_{\mathbf{F}_{t}} \log q\left(\mathbf{F}_{t} | \mathbf{F}_{0}\right) - \hat{\epsilon}_{\mathbf{F}}\left(C_{t}, t\right) \right\|_{2}^{2} \Big{]}}
\end{align}
where $\lambda_{t} = 1 / \mathbb{E}_{\mathbf{F}_{t}} \left[ \left\| \nabla_{\mathbf{F}_{t}} \log q\left(\mathbf{F}_{t} | \mathbf{F}_{0}\right) \right\|_{2}^{2} \right]$ can be approximated by Monte Carlo sampling methods. 

\noindent{\bf Diffusion on $\mathbf{L}$.} We employ the Denoising Diffusion Probabilistic Model (DDPM) \cite{ho2020denoising} to diffuse and generate $\mathbf{L}$. The forward diffusion process is defined as a Markov chain with the transition kernel $q(\mathbf{L}_{t} | \mathbf{L}_{t-1}) = \mathcal{N}(\mathbf{L}_{t}; \sqrt{1-\beta_{t}} \mathbf{L}_{t-1}, \beta_{t} \mathbf{I})$, where the variance schedule $\{\beta_{t}\}_{t=1}^{T}$ is fixed with $\beta_{t} \in (0,1)$. Let $\alpha_t = 1 - \beta_t$ and $\bar{\alpha}_t = \prod_{s=1}^{t} \alpha_s$.

This process progressively transforms $L_0$ into the normal prior distribution $p(\mathbf{L}_T) = \mathcal{N}(0, \mathbf{I})$ for a sufficiently large $T$, since $q(\mathbf{L}_t | \mathbf{L}_0) \\ = \mathcal{N}(\mathbf{L}_t; \sqrt{\alpha_t} \mathbf{L}_0, (1 - \alpha_t) \mathbf{I})$. For the reverse generation process, we model the reverse dynamics as a conditional Markov chain with learnable transitions $p_{\theta}(\mathbf{L}_{t-1} | \mathbf{L}_t) = \mathcal{N}(\mathbf{L}_{t-1}; \mu_{\theta}(\mathbf{L}_t, t), \Sigma_{\theta}(\mathbf{L}_t, t))$. Here, the mean $\mu_{\theta}$ is given by $\mu_{\theta}(\mathbf{L}_t, t) = \frac{1}{\sqrt{\alpha_t}} \left( \mathbf{L}_t - \frac{\beta_t}{\sqrt{1-\alpha_t}} \hat{\epsilon}_{\mathbf{L}}(\mathbf{L}_t, t) \right)$ and the denoising term $\hat{\epsilon}_{\mathbf{L}}$ is obtained from the denoising model $\phi_{\theta}(\mathbf{A}, \mathbf{F}_t, \mathbf{L}_t, t)$, which will be detailed later. The training loss is formulated as the expected $\ell_2$-distance between the ground truth $\epsilon_{\mathbf{L}}$ and the estimated $\hat{\epsilon}_{\mathbf{L}}$:
\begin{align}\label{equation:lossL}
\mathcal{L}_{\mathbf{L}} = \mathbb{E}_{\epsilon_{\mathbf{L}} \sim \mathcal{N}(0, \mathbf{I}), t \sim \mathcal{U}(1, T)} \left[ \| \epsilon_{\mathbf{L}} - \hat{\epsilon}_{\mathbf{L}}(\mathbf{L}_t, t) \|_2^2 \right].
\end{align}

\subsection{Denoising Model}\label{section:Denoising}

Here we introduce our proposed denoising model $\phi(\mathbf{L}, \mathbf{F}, \mathbf{A}, t)$ that produces $\hat{\epsilon}_{\mathbf{L}}$ and $\hat{\epsilon}_{\mathbf{F}}$. Let $\mathbf{H}^{(\ell)} = \left[\mathbf{h}_{1}^{(\ell)}, \mathbf{h}_{2}^{(\ell)}, \cdots, \mathbf{h}_{N}^{(\ell)}\right]$ denote the vertex representations of the $\ell$-th layer. The input feature for vertex $i$ is given by $\mathbf{h}_{i}^{(0)} = \operatorname{MLP}(f_{\text{atom}}(\mathbf{a}_{i}), f_{\text{position}}(t))$, where $f_{\text{atom}}$ and $f_{\text{position}}$ are the atomic embedding and sinusoidal positional encoding \cite{ho2020denoising}, respectively.

Inspired by the Equivariant Hypergraph Neural Networks (EHNN) framework proposed by \citet{kim2022equivariant}, we construct our model using an unfolded layer architecture. The $\ell$-th layer of this framework (EHNN-MLP) can be formulated as follows:
\begin{align}\label{equation:MP}
\left\{
\begin{array}{l}
\mathbf{m}_{ \mathbf{e} }^{(\ell)} = \operatorname{MLP}\left( \{\mathbf{h}_{j}^{(\ell-1)}\}_{j\in \mathbf{e}}, \mathbf{L}^{\top}\mathbf{L}, \psi\!\left( \sum_{j,j' \in \mathbf{e}} (f_{j'}-f_{j})\right)\right),\\
\mathbf{m}_{i}^{(\ell)} = \sum_{\mathbf{e} \ni i} \frac{\mathbf{m}_{\mathbf{e}}^{(\ell)}}{|\mathbf{e}|},\\
\mathbf{h}_{i}^{(\ell)} = \mathbf{h}_{i}^{(\ell-1)} + \operatorname{MLP}\left(\mathbf{h}_{i}^{(\ell-1)}, \mathbf{m}_{i}^{(\ell)}\right),
\end{array}
\right.
\end{align}
where $\mathbf{e}$ denotes a hyperedge, $\psi:(-1,1)^3 \to [-1,1]^{3 \times K}$ is a Fourier feature function with $K$ frequency components, and $\psi$ is periodic translation invariant.

After $\mathcal{L}$ layers of message passing conducted on the fully connected hypergraph, the lattice noise $\hat{\epsilon}_{\mathbf{L}}$ is obtained via a linear combination of $\mathbf{L}$ weighted by the output of the final layer, and the fractional coordinate score $\hat{\epsilon}_{\mathbf{F}}$ is computed column-wise. Formally:
\begin{align}\label{equation:epsilonF}
\hat{\epsilon}_{\mathbf{L}} = \mathbf{L} \varphi \left( \frac{1}{N} \sum_{i=1}^{N} \mathbf{h}_{i}^{(\mathcal{L})} \right), \quad \hat{\epsilon}_{\mathbf{F}} = \Big{[}\varphi_{\mathbf{F}}\!\left( \mathbf{h}_i^{(\mathcal{L})} \right)\Big{]}_{i=1}^N,
\end{align}
where $\varphi$ is an MLP with output in $\mathbb{R}^{3 \times 3}$, $\hat{\epsilon}_{\mathbf{F}}[:, {i}]$ denotes the $i$-th column of $\hat{\epsilon}_{\mathbf{F}}$, and $\varphi_{\mathbf{F}}$ is an MLP applied to the final vertex representation.

\begin{proposition}\label{proposition:3.1}
Propositions regarding $\hat{\epsilon}_\mathbf{L}$ and $\hat{\epsilon}_\mathbf{F}$:
\begin{itemize}
    \item $\hat{\epsilon}_\mathbf{L}$ is $O(3)$ equivariant;
    \item $\hat{\epsilon}_\mathbf{F}$ is periodic translation invariant.
\end{itemize}
Consequently, the generated distribution by EH-Diff is periodic $E(3)$ invariant.
\end{proposition}

\begin{proof}[Proof Sketch]
\textit{(i)} The $O(3)$-equivariance of $\hat{\epsilon}_{\mathbf{L}}$ follows from the invariance of the Gram matrix $\mathbf{L}^{\top}\mathbf{L}$ under orthogonal transformations and the linear dependence of the output on $\mathbf{L}$.
\textit{(ii)} The periodic translation invariance of $\hat{\epsilon}_{\mathbf{F}}$ is ensured by the Fourier features $\psi$, which use $\sin$ and $\cos$ with integer frequencies---these are invariant to integer shifts arising from the wrapping function $w(\cdot)$.
\end{proof}

\begin{proposition}\label{proposition:3.2}
If $\hat{\epsilon}_\mathbf{L}$ is $O(3)$-equivariant, then the learned marginal distribution $p_\theta(\mathbf{L}_0)$ is $O(3)$-invariant. Similarly, if $\hat{\epsilon}_\mathbf{F}$ is periodic translation invariant, the learned marginal distribution $p_\theta(\mathbf{F}_0)$ is also periodic translation invariant.
\end{proposition}

\begin{proof}[Proof Sketch]
\textit{(i) $O(3)$-invariance of $p_\theta(\mathbf{L}_0)$:} The prior $\mathcal{N}(\mathbf{0}, \mathbf{I})$ is $O(3)$-invariant. By the $O(3)$-equivariance of $\hat{\epsilon}_{\mathbf{L}}$, the reverse transition mean satisfies $\boldsymbol{\mu}_t(\mathbf{Q}\mathbf{L}_t) = \mathbf{Q}\boldsymbol{\mu}_t(\mathbf{L}_t)$, making the transition $O(3)$-equivariant. The result follows from \Cref{lemma:invariance}.
\textit{(ii) Periodic translation invariance of $p_\theta(\mathbf{F}_0)$:} The uniform prior on $[0,1)^{3 \times N}$ is periodic translation invariant. Since $\hat{\epsilon}_{\mathbf{F}}$ is also invariant, the Langevin dynamics preserves this symmetry throughout the reverse process.
\end{proof}

The full proofs of \Cref{proposition:3.1} and \Cref{proposition:3.2} can be found in \Cref{appendix:proofs}.

\begin{algorithm}[tb]
\caption{EH-Diff Training Procedure}
\label{alg:Training}
\SetAlgoLined
\KwIn{atom features $\mathbf{A}$, fractional coordinates $\mathbf{F}$, lattice $\mathbf{L}$, number of diffusion steps $T_{\text{diffusion}}$.}
Sample $\epsilon_{\mathbf{L}} \sim \mathcal{N}(0, \mathbf{I})$\;
Sample $\epsilon_{\mathbf{F}} \sim \mathcal{N}(0, \mathbf{I})$\;
Sample $t \sim \mathcal{U}(1, T_{\text{diffusion}})$\;
Update $\mathbf{L}_t \gets \sqrt{\bar{\alpha}_t} \mathbf{L}_0 + \sqrt{1 - \bar{\alpha}_t} \epsilon_{\mathbf{L}}$\;
Update $\mathbf{F}_t \gets w \left( \mathbf{F}_0 + \sigma_t \epsilon_{\mathbf{F}} \right)$\;
Update $\hat{\epsilon}_{\mathbf{L}}, \hat{\epsilon}_{\mathbf{F}} \gets \phi_{\theta} (\mathbf{A}, \mathbf{F}_t, \mathbf{L}_t, t)$, see Eq. \eqref{equation:epsilonF}\;
Minimize $\mathcal{L}_{\mathbf{L}} + \mathcal{L}_{\mathbf{F}}$, see Eq. \eqref{equation:lossF} and \eqref{equation:lossL}\;
\KwOut{Model $\phi_{\theta}$}
\end{algorithm}

\begin{table}[b]
\centering
\caption{Computational complexity comparison between EH-Diff and pairwise GNN baselines.}
\label{tab:complexity}
\begin{tabular}{lcc}
\toprule
\textbf{Method} & \textbf{Time Complexity} & \textbf{Typical Scale} \\
\midrule
Pairwise GNN& $\mathcal{O}(n^2 \cdot d^2)$ & $n^2$ edges \\
EH-Diff (ours) & $\mathcal{O}(m \cdot \bar{k} \cdot d^2)$ & $m \approx 3\text{--}5n$, $\bar{k} \approx 4$ \\
\bottomrule
\end{tabular}
\end{table}

\begin{table*}[t]
\small
\centering
\caption{{Results on Stable Structure Prediction Task.} \emph{MR} denotes the match rate. We use \textbf{boldface} to denote the best result and \underline{underline} to denote the second best result. ``-" indicates missing data in original papers.
}
\label{table:results}
\begin{tabular}{lccccccccc}
\toprule
\multirow{2}{*}{\textbf{Method}} & \multirow{2}{*}{\emph{num of Samples}} & \multicolumn{2}{c}{\textbf{Perov-5}} & \multicolumn{2}{c}{\textbf{Carbon-24}} & \multicolumn{2}{c}{\textbf{MP-20}} & \multicolumn{2}{c}{\textbf{MPTS-52}} \\
\cmidrule(lr){3-4} \cmidrule(lr){5-6} \cmidrule(lr){7-8} \cmidrule(lr){9-10}
 & & \emph{MR}\;$\uparrow$ & \emph{RMSE}\;$\downarrow$ & \emph{MR}\;$\uparrow$ & \emph{RMSE}\;$\downarrow$ & \emph{MR}\;$\uparrow$ & \emph{RMSE}\;$\downarrow$ & \emph{MR}\;$\uparrow$ & \emph{RMSE}\;$\downarrow$ \\
\midrule
\multirow{2}{*}{\textbf{RS}}
& 20 & 29.22 & 0.2924 & 14.63 & 0.4041 & 8.73 & 0.2501 & 2.05 & 0.3329 \\
& 5,000 & 36.56 & 0.0886 & 14.63 & 0.4041 & 11.49 & 0.2822 & 2.68 & 0.3444 \\
\midrule
\multirow{2}{*}{\textbf{BO}}
& 20 & 21.03 & 0.2830 & 0.44 & 0.3653 & 8.11 & 0.2402 & 2.05 & 0.3024 \\
& 5,000 & 55.09 & 0.2037 & 12.17 & 0.4089 & 12.68 & 0.2816 & 6.69 & 0.3444 \\
\midrule
\multirow{2}{*}{\textbf{PSO}}
& 20 & 20.90 & 0.0836 & 6.40 & 0.4204 & 4.05 & 0.1567 & 1.06 & 0.2339 \\
& 5,000 & 21.88 & 0.0844 & 6.50 & 0.4211 & 4.35 & 0.1670 & 1.09 & 0.2390 \\
\midrule
\midrule
\multirow{1}{*}{\textbf{P-cG-SchNet}}
& 1 & 48.22 & 0.4179 & 17.29 & 0.3846 & 15.39 & 0.3762 & 3.67 & 0.4115 \\
\midrule
\multirow{1}{*}{\textbf{CDVAE}}
& 1 & 45.31 & 0.1138 & 17.09 & 0.2969 & 33.90 & 0.1045 & 5.34 & 0.2106 \\
\midrule
\multirow{1}{*}{\textbf{EquiCSP}}
& 1 & 52.02 & \underline{0.0707} & - & - &\underline{57.59} &\underline{0.0510} &\underline{14.85} &\textbf{0.1169} \\
\midrule
\multirow{1}{*}{\textbf{DiffCSP}}
& 1 & 52.02 & 0.0760 & \underline{17.54} & 0.2759 & 51.49 & 0.0631 & 12.19 & 0.1786 \\
\midrule
\midrule
 \multirow{1}{*}{\textbf{EH-Diff} (Sphere)} 
& 1 
& \textbf{54.59} 
& 0.0728
& {17.28} 
& \textbf{0.2623}
& 56.81 
& 0.0559
& 13.97 
& \underline{0.1285}\\
\midrule
 \multirow{1}{*}{\textbf{EH-Diff} (Cube)} 
& 1 
& \underline{53.43} 
& \textbf{0.0706}
& \textbf{18.11} 
& \underline{0.2757}
& \textbf{58.12} 
& \textbf{0.0460}
& \textbf{16.99} 
& 0.1340\\ 
\bottomrule
\end{tabular}
\end{table*}

\textbf{Computational Complexity.} Let $n$ denote the number of atoms and $m$ the number of hyperedges. The message passing in \Cref{equation:MP} has time complexity $\mathcal{O}(m \cdot \bar{k} \cdot d^2)$, where $\bar{k}$ is the average hyperedge size and $d$ is the hidden dimension. For comparison, pairwise GNNs operating on fully connected graphs have complexity $\mathcal{O}(n^2 \cdot d^2)$. In practice, our method achieves comparable or better efficiency since $m \ll n^2$; empirically, we observe $m \approx 3\text{--}5n$ across all datasets. \Cref{tab:complexity} provides a detailed comparison.

\subsection{CSP Workflow}

{\bf Training via Hypergraph Diffusion.} \Cref{alg:Training} outlines the training procedure for the EH-Diff model. The process begins by taking as input the atom features, fractional coordinates, and lattice of the crystal, along with the number of diffusion steps. Initially, noise is sampled for both the lattice and the fractional coordinates from a normal distribution. A random diffusion step is selected from a uniform distribution. The lattice and coordinates are then updated using their respective diffusion equations, incorporating noise to approximate their final states. Next, the model parameters are updated using a denoising function, and the objective is to minimize the combined loss functions for the lattice and fractional coordinates. This iterative training process continues until the model effectively captures the desired crystal properties, thus providing a learned representation capable of predicting complex crystal structures.

\begin{algorithm}[tb]
\caption{EH-Diff Inference Procedure}
\label{alg:Inference}
\SetAlgoLined
\KwIn{atom features $\mathbf{A}$, number of diffusion steps $T_{\text{diffusion}}$, trained model $\phi_\theta$, hyperparameters $\bar{\alpha}_t, \tilde{\beta}_t$.}
Sample $\mathbf{L}_T \sim \mathcal{N}(\mathbf{0}, \mathbf{I})$, $\mathbf{F}_T \sim \mathcal{U}(\mathbf{0}, 1)$\;
\For{$t = T_{\text{diffusion}}$ \KwTo $1$}{
    \If{$t > 1$}{
        Sample $\epsilon_{\mathbf{L}} \sim \mathcal{N}(\mathbf{0}, \mathbf{I})$, $\epsilon_{\mathbf{F}} \sim \mathcal{N}(\mathbf{0}, \mathbf{I})$\;
    }
    $\hat{\epsilon}_{\mathbf{L}}, \hat{\epsilon}_{\mathbf{F}} \gets \phi_\theta (\mathbf{A}, \mathbf{F}_t, \mathbf{L}_t, t)$ via trained model $\phi_\theta$\;
    Update $\mathbf{L}_{t-1}$ via reverse DDPM using $\hat{\epsilon}_{\mathbf{L}}$, $\bar{\alpha}_t$, $\tilde{\beta}_t$\;
    Update $\mathbf{F}_{t-1} \gets w\!\left(\mathbf{F}_t + \eta_t \hat{\epsilon}_{\mathbf{F}} + \sqrt{2\eta_t}\,\epsilon_{\mathbf{F}}\right)$ via predictor-corrector\;
}
\KwOut{Predicted lattice $\hat{\mathbf{L}}_0$ and fractional coordinates $\hat{\mathbf{F}}_0$}
\end{algorithm}

{\bf Inference.} \Cref{alg:Inference} presents the inference procedure for the EH-Diff model. Given the atom features $\mathbf{A}$, the lattice is initialized from a Gaussian prior and the fractional coordinates from a uniform prior. The reverse diffusion then proceeds iteratively from step $T$ to $1$. At each step, the trained denoising model $\phi_\theta$ predicts the noise components $\hat{\epsilon}_{\mathbf{L}}$ and $\hat{\epsilon}_{\mathbf{F}}$, which are used to update $\mathbf{L}_t$ via the reverse DDPM transition and $\mathbf{F}_t$ via the predictor-corrector method with wrapping. The final outputs are the predicted lattice $\hat{\mathbf{L}}_0$ and fractional coordinates $\hat{\mathbf{F}}_0$.

\section{Experiments}

\subsection{Experimental Setup}

{\bf Dataset.} We conducted experiments on four datasets, each with distinct levels of complexity. The first dataset, Perov-5, comprises 18,928 perovskite materials with structurally similar configurations, containing 5 atoms per unit cell \cite{castelli2012new, castelli2012computational}. The second dataset, Carbon-24, includes 10,153 carbon-based materials, with unit cells ranging from 6 to 24 atoms. The third dataset, MP-20, is derived from the Materials Project and contains 45,231 stable inorganic materials, which include most of the experimentally synthesized compounds with up to 20 atoms per unit cell \cite{jain2013commentary}. Lastly, the MPTS-52 dataset represents a more challenging extension of MP-20, consisting of 40,476 structures with up to 52 atoms per unit cell, ordered chronologically by their earliest publication year. For the Perov-5, Carbon-24 and MP-20 data sets, we adhered to the train/validation/test split ratio of 0.8/0.1/0.1 as suggested by \citet{xie2021crystal}. For the MPTS-52 dataset, a chronological split of 1/2/3 was used for training, validation, and testing, respectively.

\noindent{\bf Baseline.} This study distinguishes itself from previous research by focusing on a comparison of two main methodologies. The first methodology follows a sequential predict-and-optimize strategy. In this approach, a model is initially trained to predict material properties, and then optimization algorithms are employed to determine the most favorable structures. We adapt the approach described by \citet{cheng2022crystal}, using MEGNet \cite{chen2019graph} to predict formation energies. To convert these predictions into optimal structures, we apply three optimization techniques: random search (RS), Bayesian optimization (BO), and particle swarm optimization (PSO), each of which is run for 5,000 iterations to thoroughly explore the solution space. The second methodology utilizes generative models. Building on the enhancements proposed by \citet{xie2018crystal}, we leverage the cG-SchNet framework \cite{gebauer2022inverse}, which incorporates SchNet \cite{schutt2018schnet} as its fundamental model, along with a ground-truth lattice initialization to capture periodic information, resulting in the P-cG-SchNet model. We also explore conditional diffusion-based generative models such as CDVAE \cite{xie2021crystal}, EquiCSP \cite{lin2024equicsp} and DiffCSP \cite{jiao2024crystal}.

\begin{figure*}[t]
    \centering
    \includegraphics[width=\linewidth]{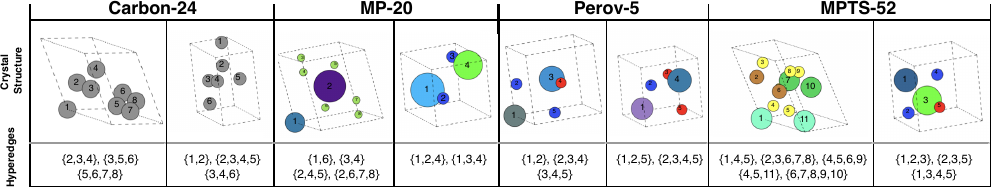}
    \caption{Visualization of the hypergraph-based crystal structure for each dataset with the corresponding hyperedges of order greater than one.}
    \label{fig:demo}
\end{figure*}

\noindent{\bf Evaluation Metrics.} We evaluate the performance of different methods by reporting the Match Rate (\emph{MR}) and Root Mean Square Error (\emph{RMSE}), as shown in \Cref{table:results}. Following established protocols \cite{xie2021crystal, jiao2024crystal}, the assessment involves comparing predicted candidates with ground truth structures. Specifically, for each structure in the test set, we generate $k$ samples of the same chemical composition and determine a match if at least one sample aligns with the ground-truth structure. This matching is carried out using the StructureMatcher class from pymatgen \cite{ong2013python}, with the thresholds set at $\text{stol} = 0.5$, $\text{angle}_{\text{tol}} = 10$, and $\text{ltol} = 0.3$\footnote{These tolerance parameters control the strictness of structural matching in StructureMatcher: \text{stol} (site tolerance) defines the maximum permissible distance between atomic positions after optimal alignment, \text{angletol} (angle tolerance) sets the maximum angular deviation between lattice vectors, and \text{ltol} (length tolerance) specifies the maximum relative difference in lattice parameter magnitudes. Structures are considered matched only if all three criteria are simultaneously satisfied.}. The \emph{MR} is defined as the ratio of successfully matched structures to the total number of structures in the test set. The \emph{RMSE} is calculated between the ground-truth structure and the best-matching candidate, normalized by $\sqrt{\frac{V}{n}}$, where $V$ denotes the lattice volume, and averaged across all matched structures. For optimization methods, we select the 20 lowest-energy structures from the 5,000 generated during testing as the candidates. In contrast, for generative baselines and our EH-Diff approach, we evaluate performance using $k = 1$.

\noindent{\bf Hypergraph Representation of Crystals.} We formalize the crystal-to-hypergraph mapping as follows. Each atom in the unit cell corresponds to a vertex $v \in \mathcal{V}$, and each hyperedge $e \in \mathcal{E}$ encodes a group of atoms that participate in a collective interaction---analogous to cliques or stars in conventional graph representations. This formulation naturally captures multi-body interactions that arise in crystalline materials. \Cref{fig:demo} illustrates representative examples from each dataset. The hypergraph representation is characterized by three key components:

\noindent(1) \textit{Vertices}: Each vertex represents an atom within the unit cell, encoding atomic features such as element type and oxidation state.\\
(2) \textit{Hyperedges}: Each hyperedge connects a subset of atoms involved in collective interactions, including coordination polyhedra, hydrogen bonding networks, and van der Waals clusters.\\
(3) \textit{Edge weights} (optional): Hyperedges can be weighted to reflect interaction strength, assigning higher weights to covalent bonds than to weak dispersive forces.

\noindent{\bf Hyperedge Construction.} The modeling of hyperedges in crystal structures is illustrated in \Cref{fig:EH-Diff-hypergraph}. We employ two geometric primitives---spheres and cubes---to scan localized regions of space and establish hyperedges. A hyperedge is formed when multiple atoms fall within the same scanning volume.

\textit{Sphere-based hyperedges} (left panel) utilize a spherical scanning region centered on each atom. This approach naturally captures isotropic coordination shells (\emph{e.g.}, the first coordination sphere) and provides flexibility for non-uniform atomic distributions.

\textit{Cube-based hyperedges} (right panel) employ cubic scanning regions, which better align with crystallographic symmetry in cubic crystal systems and offer computational efficiency due to their axis-aligned geometry.

The choice between these two methods depends on the specific characteristics of the crystal structure. The scanning radius $r$ is determined by typical bond lengths, set as $r = 1.5 \times \bar{d}_{\text{NN}}$, where $\bar{d}_{\text{NN}}$ denotes the average nearest-neighbor distance in the training set.

\begin{remark}[Symmetry and Periodicity]
Crystals have inherent symmetry and periodicity, which can be modeled in a hypergraph. The repeating unit cell of a crystal can be represented by a subgraph, with vertices for atoms and hyperedges for atomic bond. By applying symmetry operations such as rotations and reflections, the hypergraph can be transformed to generate the entire crystal structure from this subgraph, effectively capturing both the periodicity and symmetry of the crystal in a compact form.
\end{remark}

\begin{table}[t]
\small
\centering
\caption{Overall results over the 15 selected compounds in MP-20.}
\label{table:selectedresults}
\begin{tabular}{lccc}
\toprule
\textbf{Method} & \emph{MR}\;$\uparrow$ & \emph{Avg.RMSE}\;$\downarrow$ & \emph{Avg.Time}\;$\downarrow$ \\
\midrule
USPEX & 53.33 & 0.0159 & 12.5h \\
\midrule
DiffCSP & 73.33 & 0.0172 & \textbf{10s} \\
\midrule\midrule
  EH-Diff (Sphere) 
&  \textbf{82.69} 
&  \textbf{0.0092} 
&  13 $\sim$ 14s\\
\midrule
  EH-Diff (Cube) 
&  79.77 
&  0.0109 
&  13 $\sim$ 14s\\
\bottomrule
\end{tabular}
\end{table}

\begin{table*}[t]
\small
\centering
\caption{\emph{Results on ab initio generation task.} We use \textbf{boldface} to denote the best result and \underline{underline} to denote the second best result. We also use \textcolor{red}{\textbf{red}} to denote the best rank score. ``-" indicates missing data in original papers.
}\label{tab:resultsabinitio}
\begin{tabular}{clcccccccl}
\toprule
\multirow{2}{*}{\textbf{Data}} & \multirow{2}{*}{\textbf{Method}} & \multicolumn{2}{c}{\textbf{Validity(\%)}\;$\uparrow$} & \multicolumn{2}{c}{\textbf{Coverage(\%)}\;$\uparrow$} & \multicolumn{3}{c}{\textbf{Property}\;$\downarrow$} 
&\textbf{Rank}\;\multirow{2}{*}{$\downarrow$}\\
\cmidrule(lr){3-4}
\cmidrule(lr){5-6}
\cmidrule(lr){7-9}
& & Structure & Composition & {COV-Recall}& {COV-Precision} & $d_{\rho}$ & $d_{E}$ & $d_{\text{elem}}$ &\textbf{Score} \\
\midrule
\multirow{10}{*}{\rotatebox{90}{\textbf{Perov-5}}} & FTCP & 0.24 & 54.24 & 0.00 & 0.00 & 10.27 & 156.0 & 0.6297 & 9.9\\
& Cond-DFC-VAE & 73.60 & 82.95 & 73.92 & 10.13 & 2.268 & 4.111 & 0.8373 & 8.4\\
& G-SchNet & 99.92 & 98.79 & 0.18 & 0.23 & 1.625 & 4.746 & 0.0368 & 7.0\\
& P-G-SchNet & 79.63 & \textbf{99.13} & 0.37 & 0.25 & 0.2755 & 1.388 & 0.4552 & 6.7\\
& CDVAE & \textbf{100.0} & 98.59 & 99.45 & 98.46 & 0.1258 & 0.0264 & 0.0628 & 5.0\\
& SyMat & \textbf{100.0} & 97.40 & 99.68 & 98.64 & 0.1893 & 0.2364 & 0.0177 & 4.4\\
& EquiCSP & \textbf{100.0} & 98.60 & 99.60 & \underline{98.76} & \underline{0.1110} & \underline{0.0257} & 0.0503 & 3.3\\
& DiffCSP & \textbf{100.0} & \underline{98.85} & \underline{99.74} & 98.27 & \underline{0.1110} & 0.0263 & \underline{0.0128} & 2.6\\
\cmidrule{2-9}
& EH-Diff (Sphere)
& \textbf{100.0} & 98.60 & 99.68 & \textbf{98.82} & 0.1116 & 0.0879 & \textbf{0.0118} & 2.9
\\
& EH-Diff (Cube) 
& \textbf{100.0} & {98.75} & \textbf{99.77} & \textbf{98.82} & \textbf{0.1103} & \textbf{0.0214} & 0.0620 & \textcolor{red}{\textbf{2.1}}
\\
\midrule
\midrule
\multirow{9}{*}{\rotatebox{90}{\textbf{Carbon-24}}} & FTCP & 0.08 & - & 0.00 & 0.00 & 5.206 & 19.05 & - &8.0\\
& G-SchNet & \underline{99.94} & - & 0.00 & 0.00 & 0.9427 & 1.320 & - &7.0\\
& P-G-SchNet & 48.39 & - & 0.00 & 0.00 & 1.533 & 134.7 & - &7.8\\
& CDVAE & \textbf{100.0} & - & 99.80 & 83.08 & 0.1407 & 0.2850 & - &4.6\\
& SyMat & \textbf{100.0} & - & \textbf{100.0} & 97.59 & 0.1195 & 3.9576 & - &3.2\\
& EquiCSP & \textbf{100.0} & - & 99.75 & 97.12 & 0.0734 & \textbf{0.0508} & - &3.2\\
& DiffCSP & \textbf{100.0} & - & \underline{99.90} & 97.27 & 0.0805 & 0.0820 & - &3.0\\
\cmidrule{2-9}
& EH-Diff (Sphere) 
& \textbf{100.0} & - & \underline{99.90} & \textbf{99.13} & \underline{0.0655} & 0.0561 & -
& 1.8
\\
& EH-Diff (Cube) 
& \textbf{100.0} & - & \underline{99.90} & \underline{98.03} & \textbf{0.0601} & \underline{0.0533} & -
& \textcolor{red}{\textbf{1.6}}
\\
\midrule
\midrule
\multirow{9}{*}{\rotatebox{90}{\textbf{MP-20}}} & FTCP & 1.55 & 48.37 & 4.72 & 0.09 & 23.71 & 160.9 & 0.7363 &8.9\\
& G-SchNet & 99.65 & 75.96 & 38.33 & 99.57 & 3.034 & 42.09 & 0.6411 &7.4\\
& P-G-SchNet & 77.51 & 76.40 & 41.93 & 99.74 & 4.04 & 2.448 & 0.6234 &6.9\\
& CDVAE & \textbf{100.0} & 86.70 & 99.15 & 99.49 & 0.6875 & 0.2778 & 1.432 &5.3\\
& SyMat & \textbf{100.0} & \underline{88.26} & {98.97} & \textbf{99.97} & 0.3805 & 0.3506 & 0.5067 &3.6\\
& EquiCSP & 99.97 & 82.20 & 99.65 & 99.68 & \textbf{0.1300} & 0.0848 & 0.3978 &4.1\\
& DiffCSP & \textbf{100.0} & 83.25 & \underline{99.71} & 99.76 & 0.3502 & 0.1247 & \textbf{0.3398} &2.7\\
\cmidrule{2-9}
& EH-Diff (Sphere) 
& \underline{99.90} & \textbf{89.50} & 99.69 & \underline{99.89} & 0.3996 & \textbf{0.0590} & 0.4119
& 3.0
\\
& EH-Diff (Cube) 
& \textbf{100.0} & 82.50 & \textbf{99.81} & {99.86} & \underline{0.2880} & \underline{0.0804} & \underline{0.3474}
& \textcolor{red}{\textbf{2.3}}
\\
\bottomrule
\end{tabular}
\end{table*}

\subsection{Main Results}

\Cref{table:results} presents the performance of EH-Diff (Sphere and Cube) on the Stable Structure Prediction task across four datasets. EH-Diff consistently achieves the best or second-best Match Rate and RMSE even with a single diffusion sample, outperforming both optimization-based methods (RS, BO, PSO) that require thousands of samples and other generative models such as CDVAE and EquiCSP. Specifically, on the more complex MP-20 and MPTS-52 datasets, EH-Diff (Cube) achieves the highest MR (58.12 and 16.99, respectively) and the lowest RMSE (0.0460 and 0.1340).

\Cref{table:selectedresults} further shows that we achieve superior accuracy with competitive runtimes. EH-Diff (Sphere) has the highest MR (82.69) and lowest RMSE (0.0092), while maintaining an efficient runtime of 13$\sim$14 seconds, comparable to DiffCSP's 10 seconds. This indicates that EH-Diff can achieve high prediction accuracy without sacrificing efficiency.

\subsection{Results with Ab Initio Crystal Generation}

Although our method is initially designed to tackle CSP with fixed composition $\mathbf{A}$, it can be extended to address the ab initio generation task by also generating $\mathbf{A}$. This extension is achieved by optimizing the one-hot representation of $\mathbf{A}$ using a DDPM-based approach. More details can be found in \Cref{appendix:abinitio}\footnote{Composition-based metrics are irrelevant for Carbon-24, as its structures consist solely of carbon. Thus, the corresponding results in \Cref{tab:resultsabinitio} are represented by a ``-".}. In addition, we define a ranking score for each model as $``\text{Rank Score}" = \frac{\sum_{m \in \mathcal{M}} \text{rank}_m }{|\mathcal{M}|}$, where $\mathcal{M}$ represents the set of metrics and $\text{rank}_m$ denotes the rank of the model according to the metric $m$.

\Cref{tab:resultsabinitio} presents a comprehensive evaluation of ab initio crystal generation methods. Both EH-Diff variants achieve 100\% structural validity across all datasets, matching or surpassing state-of-the-art methods (DiffCSP, CDVAE, EquiCSP). Both variants also exhibit outstanding coverage metrics (COV-Recall/Precision), with near-perfect scores on Perov-5 and Carbon-24. EH-Diff (Cube) additionally excels in property accuracy metrics (\(d_\rho\), \(d_E\), \(d_{\text{elem}}\)), indicating high physical and chemical plausibility of the generated structures. Notably, EH-Diff (Cube) achieves the best overall rank score across all datasets, underscoring the benefits of capturing high-order interactions and preserving crystal symmetries. Overall, the Cube variant consistently outperforms the Sphere variant across most metrics.

\subsection{Ablation Study}

To investigate the contribution of high-order interactions, we conduct an ablation study by varying the maximum hyperedge order $k$ in our model. Specifically, we train EH-Diff variants with $k \in \{2, 3, 4, 5\}$ on the MP-20 dataset, where $k=2$ corresponds to the pairwise-only baseline (equivalent to a standard graph representation). The results are summarized in \Cref{tab:ablation_order} and we summarize the key observations below:

\begin{table}[t]
\centering
\caption{Ablation study on maximum hyperedge order $k$ using the MP-20 dataset.}
\label{tab:ablation_order}
\begin{tabular}{lcc}
\toprule
\textbf{Max Hyperedge Order} & \textbf{MR (\%) $\uparrow$} & \textbf{RMSE $\downarrow$} \\
\midrule
$k=2$ (pairwise only) & 51.49 & 0.0631 \\
$k=3$ & 55.23 & 0.0512 \\
$k=4$ & 57.81 & 0.0475 \\
$k=5$ & \textbf{58.12} & \textbf{0.0460} \\
\bottomrule
\end{tabular}
\end{table}

\noindent(1) \textit{High-order interactions are essential.} Increasing the hyperedge order from $k=2$ (pairwise) to $k=5$ (full) yields a substantial improvement of \textbf{6.63\%} in Match Rate (from 51.49\% to 58.12\%) and a \textbf{27.1\%} reduction in RMSE (from 0.0631 to 0.0460). This confirms that multi-atomic interactions beyond pairwise bonds carry critical structural information for accurate crystal structure prediction.\\
(2) \textit{Diminishing returns at higher orders.} The performance gain is most pronounced when transitioning from $k=2$ to $k=3$ (+3.74\% MR), capturing essential three-body interactions such as bond angles. Subsequent increases in $k$ yield progressively smaller improvements, suggesting that the majority of high-order structural information is encoded in hyperedges of order 3--4.\\
(3) \textit{Alignment with coordination chemistry.} The observed trend aligns with physical intuition: coordination polyhedra in crystals typically involve 4--6 atoms (\emph{e.g.}, tetrahedral and octahedral coordination), which are effectively captured by hyperedges of order $k \geq 4$. The marginal gain from $k=4$ to $k=5$ indicates that most coordination environments in MP-20 are adequately represented with $k=4$.

These findings validate our hypergraph-based approach and demonstrate that explicitly modeling multi-atomic interactions leads to significant performance gains over traditional pairwise graph representations.

\section{Conclusion}

We introduce a hypergraph-based framework for modeling crystal structures and propose the Equivariant Hypergraph Diffusion Model (EH-Diff) for Crystal Structure Prediction (CSP). By representing atomic interactions through hyperedges rather than pairwise edges, EH-Diff captures the intricate higher-order interactions inherent in crystalline materials while naturally encoding essential symmetries and periodicities. Our extensive experimental evaluation demonstrates that EH-Diff achieves state-of-the-art performance across multiple benchmark datasets, outperforming existing CSP methods even with a single diffusion sample.

\section{Limitations and Ethical Considerations}

Our approach has several limitations. First, the hyperedge construction relies on geometric heuristics (sphere/cube scanning), which may not optimally capture interactions in all crystal systems; adaptive or learnable hyperedge construction strategies could further improve performance. Second, our model is designed for periodic crystalline materials and may not directly generalize to amorphous or non-periodic structures. Third, the current framework assumes a fixed number of atoms within a unit cell during CSP inference, and extending it to variable-composition search remains an open direction.

Regarding ethical considerations, we used publicly available datasets for all experiments. We acknowledge that model performance may reflect biases present in the training data. While this research aims to accelerate materials discovery for beneficial applications, we recognize the potential for unintended misuse in designing harmful materials. We encourage responsible development and application of such generative tools, emphasizing the need for ethical oversight and experimental validation of all AI-proposed structures.

\begin{acks}
This work was partially supported by the Strategic Priority Research Program of the Chinese Academy of Sciences (No. XDB0680101), the National Natural Science Foundation of China (No. 62472416, 62376064 and 62402491), and the CAS Project for Young Scientists in Basic Research (No. YSBR-008). The model training was performed on the robotic AI-Scientist platform of Chinese Academy of Science.
\end{acks}

\bibliographystyle{ACM-Reference-Format}
\balance
\bibliography{mainref}


\appendix

\section*{Appendix}\label{sec:appendix}

\section{Preliminaries}\label{appendix:Preliminaries}

Let \( \{x_1, x_2, \dots\} \) denote a set and \( (x_1, x_2, \dots) \) denote an ordered tuple. The notation \( [n] \) denotes the set \( \{1, 2, \dots, n\} \). The space of order-\(k\) tensors is denoted by \( \mathbb{R}^{n^k \times d} \), where \(d\) is the feature dimensionality. For an order-\(k\) tensor \( \mathbf{A} \in \mathbb{R}^{n^k \times d} \), a multi-index \( \mathbf{i} = (i_1, \ldots, i_k) \in [n]^k \) indexes the element \( \mathbf{A}_{\mathbf{i}} = \mathbf{A}_{i_1, \ldots, i_k} \in \mathbb{R}^d \). Let \( \mathbb{S}_n \) denote the set of all permutations of \( [n] \). A permutation \( \pi \in \mathbb{S}_n \) acts on a multi-index \( \mathbf{i} \) by permuting its components, i.e., \( \pi(\mathbf{i}) = (\pi(i_1), \ldots, \pi(i_k)) \). Its action on the tensor \( \mathbf{A} \) is defined by permuting the tensor indices, such that the transformed tensor satisfies \( (\pi \cdot \mathbf{A})_{\mathbf{i}} = \mathbf{A}_{\pi^{-1}(\mathbf{i})} \).

\vspace{1mm}{\bf Hypergraph}. A hypergraph generalizes a graph by allowing edges to connect any subset of nodes, not just pairs. Formally, a hypergraph \( \mathbf{H} = (\mathbf{V}, \mathbf{E}, \mathbf{X}) \) consists of a node set \( \mathbf{V} \) of size \( n \), a hyperedge set \( \mathbf{E} \) of size \( m \), and a feature matrix \( \mathbf{X} \in \mathbb{R}^{m \times d} \) that encodes the attributes of the hyperedges. Each hyperedge \( e \in \mathbf{E} \) is a subset of \( \mathbf{V} \), and its size \( |e| \) denotes its order (or cardinality). For example, a first-order edge \( \{i\} \) corresponds to the \( i \)-th node itself, a second-order edge \( \{i, j\} \) represents a pairwise link, and an order-\( k \) edge \( \{i_1, \dots, i_k\} \) connects \( k \) nodes. The feature associated with hyperedge \( e \) is denoted by \( \mathbf{X}_{e} \in \mathbb{R}^d \). For theoretical analysis, hypergraphs can also be represented using higher-order tensors as \( \mathbf{H} = (\mathbf{V}, \mathbf{A}) \), where \( \mathbf{A} \in \mathbb{R}^{n^k \times d} \) is a tensor that encodes hyperedge features \citep{kim2022equivariant}. In this tensor representation, \( \mathbf{A}_{i_1, \dots, i_k} \) represents the feature of the hyperedge \( (i_1, \dots, i_k) \). While this dense tensor formulation provides a unified framework for equivariance analysis, practical implementations typically employ the sparse set-based representation \( (\mathbf{V}, \mathbf{E}, \mathbf{X}) \) where each hyperedge is uniquely defined by a set of node indices.

\vspace{1mm}{\bf Hypergraph as a Sequence of Higher-order Tensors.} 
For a hypergraph \( (\mathbf{V}, \mathbf{E}, \mathbf{X}) \) with maximum hyperedge order \( K \), we can decompose it into a sequence of \( k \)-uniform hypergraphs for all \( k \leq K \), denoted as 
\[
(\mathbf{V}, \mathbf{E}^{(k)}, \mathbf{X}^{(k)})_{k \leq K} = (\mathbf{V}, \mathbf{E}^{(:K)}, \mathbf{X}^{(:K)}),
\]
where \( \mathbf{E}^{(k)} \) is the set of all order-\( k \) hyperedges in \( \mathbf{E} \), and \( \mathbf{X}^{(k)} \) is a matrix formed by stacking the features \( \{ \mathbf{X}_e \mid e \in \mathbf{E}^{(k)} \} \) for each hyperedge \( e \) in \( \mathbf{E}^{(k)} \). This representation allows flexible modeling of heterogeneous interaction orders within a single structure.

\vspace{1mm}{\bf Permutation Equivariance.} 
In (hyper)graph learning, a common goal is to design a function \( f \) that maps a higher-order tensor \( \mathbf{A} \) to an output. Because the tensor representation changes under node permutation, it is essential for \( f \) to exhibit invariance or equivariance with respect to these permutations. If the output is a scalar or a vector, \( f \) must be permutation-invariant, meaning that
\[
f(\pi \cdot \mathbf{A}) = f(\mathbf{A})
\]
holds for any permutation \( \pi \in \mathbb{S}_n \). Conversely, if the output is a tensor, \( f \) must be permutation-equivariant, ensuring that
\[
f(\pi \cdot \mathbf{A}) = \pi \cdot f(\mathbf{A})
\]
for all \( \pi \in \mathbb{S}_n \) \citep{kim2022equivariant}. When \( f \) is implemented as a neural network (typically a sequence of linear transformations and non-linear activations), the design reduces to constructing linear layers that are invariant or equivariant under permutations.

\section{Ab Initio Generation}\label{appendix:abinitio}

We extend EH-Diff to the domain of ab initio crystal generation by incorporating discrete diffusion on atom types, denoted as $\mathbf{A}$. We compare EH-Diff against five state-of-the-art generative methods: FTCP \cite{ren2022invertible}, Cond-DFC-VAE \cite{court20203}, G-SchNet \cite{gebauer2022inverse}, its periodic extension P-G-SchNet, and the original version of the generative framework \cite{xie2021crystal}. For EH-Diff, it samples the number of atoms according to the statistical distribution estimated from the training dataset, following the approach of \citet{hoogeboom2022equivariant}. This enables the generation of crystal structures with variable sizes. Consistent with the evaluation protocol in \cite{xie2021crystal}, we assess generation performance using three metrics: Validity, Coverage, and Property Statistics. Validity measures the percentage of generated crystals that are physically plausible. Coverage quantifies the similarity between the generated samples and the test dataset in terms of structural diversity. Property statistics evaluate physical and chemical properties of the generated crystals, including density, formation energy, and the distribution of elemental compositions.

\subsection{Discrete Diffusion for Atomic Types in Ab Initio Generation}

For ab initio crystal generation, we extend EH-Diff to jointly generate atomic types $\mathbf{A}$ alongside lattice vectors $\mathbf{L}$ and fractional coordinates $\mathbf{F}$. We employ a discrete diffusion process following the D3PM framework~\cite{austin2021structured}.

\vspace{1mm}\textbf{Forward Process.} Let $\mathbf{a}_i \in \{1, \ldots, K\}$ denote the atomic type of atom $i$, where $K$ is the number of element types. The forward diffusion corrupts atomic types via a categorical transition matrix:
\begin{equation}
q(\mathbf{a}_i^t | \mathbf{a}_i^{t-1}) = \text{Cat}(\mathbf{a}_i^t; \mathbf{Q}_t \mathbf{a}_i^{t-1}),
\end{equation}
where $\mathbf{Q}_t \in \mathbb{R}^{K \times K}$ is defined as:
\begin{equation}
\mathbf{Q}_t = (1 - \beta_t) \mathbf{I} + \beta_t \mathbf{1}\mathbf{m}^\top,
\end{equation}
with $\mathbf{m} \in \mathbb{R}^K$ representing the marginal distribution of element types in the training set. This ensures that as $t \to T$, the atomic type distribution converges to the training marginal $\mathbf{m}$.

\vspace{1mm}\textbf{Reverse Process.} The reverse transition is parameterized as:
\begin{equation}
p_\theta(\mathbf{a}_i^{t-1} | \mathbf{a}_i^t, \mathbf{F}_t, \mathbf{L}_t) = \text{Cat}(\mathbf{a}_i^{t-1}; \hat{\mathbf{p}}_\theta(\mathbf{a}_i^t, \mathbf{h}_i^{(\mathcal{L})}, t)),
\end{equation}
where $\hat{\mathbf{p}}_\theta$ is an MLP that takes the final-layer node representation $\mathbf{h}_i^{(\mathcal{L})}$ (which encodes both geometric and compositional information via hypergraph message passing) and outputs a probability distribution over element types.

\vspace{1mm}\textbf{Training Loss.} The discrete diffusion loss is:
\begin{equation}
\mathcal{L}_A = \mathbb{E}_{q(\mathbf{A}_t | \mathbf{A}_0), t \sim \mathcal{U}(1,T)} \left[ -\sum_{i=1}^{N} \log p_\theta(\mathbf{a}_i^0 | \mathbf{a}_i^t, \mathbf{F}_t, \mathbf{L}_t) \right].
\end{equation}

\vspace{1mm}\textbf{Joint Training.} The total loss combines continuous and discrete components:
\begin{equation}
\mathcal{L}_{\text{total}} = \mathcal{L}_L + \mathcal{L}_F + \lambda_A \mathcal{L}_A,
\end{equation}
where $\lambda_A = 0.1$ balances the contribution of atomic type prediction.

\vspace{1mm}\textbf{Coupling with Geometry.} A key advantage of our hypergraph formulation is that atomic type prediction benefits from high-order geometric context. The hypergraph message passing aggregates information from entire coordination polyhedra, allowing the model to infer element types based on local geometric signatures (e.g., octahedral vs. tetrahedral coordination often correlates with different oxidation states and elements).

\begin{table}[t]
\centering
\caption{Comparison of high-order structural representations on the MP-20 dataset. 
All methods employ identical diffusion processes; only the structural encoding differs.}
\label{tab:representation_comparison}
\begin{tabular}{lccc}
\toprule
\textbf{Representation} & \textbf{MR (\%) $\uparrow$} & \textbf{RMSE $\downarrow$} & \textbf{Time (s) $\downarrow$} \\
\midrule
Pairwise Graph & 51.49 & 0.0631 & 10 \\
Subgraph-3 & 53.87 & 0.0573 & 18 \\
Subgraph-4 & 54.12 & 0.0561 & 47 \\
Motif-based & 52.34 & 0.0598 & 15 \\
\midrule
\textbf{Hypergraph (Ours)} & \textbf{58.12} & \textbf{0.0460} & 14 \\
\bottomrule
\end{tabular}
\end{table}

\section{Hypergraph Modeling Variants}

\begin{figure}[b]
    \centering
    \includegraphics[width=.6\linewidth]{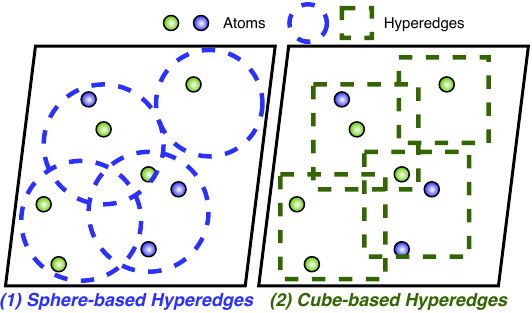}
    \caption{Illustration of hyperedge construction via spherical and cubic scanning volumes.}\label{fig:EH-Diff-hypergraph}
\end{figure}

Figure \ref{fig:EH-Diff-hypergraph} illustrates hyperedge construction in crystal structures via spherical and cubic scanning volumes. To identify hyperedges, a scanning volume (sphere or cube) is centered on or positioned to enclose specific atoms. A hyperedge is formed when a designated group of atoms (green and blue nodes) is fully contained within the same scanning volume. The left panel shows hyperedge formation using a spherical scanning volume. The right panel shows the analogous process using a cubic scanning volume. This comparison highlights how the geometry of the scanning volume influences hyperedge identification.

\section{Proofs of \Cref{proposition:3.1} and \Cref{proposition:3.2}}\label{appendix:proofs}

In this section, we provide proofs for \Cref{proposition:3.1} and \Cref{proposition:3.2}.

\subsection{Proof of \Cref{proposition:3.1}}

\begin{proof}
We prove the two claims of \Cref{proposition:3.1} separately.

\vspace{1mm}\textbf{Part I: $\hat{\epsilon}_{\mathbf{L}}$ is $O(3)$-equivariant.} We first establish the $O(3)$-invariance of the Gram matrix $\mathbf{L}^{\top}\mathbf{L}$. For any orthogonal transformation $\mathbf{Q} \in O(3)$ satisfying $\mathbf{Q}^{\top}\mathbf{Q} = \mathbf{I}$, we have:
\begin{equation}
    (\mathbf{QL})^{\top}(\mathbf{QL}) = \mathbf{L}^{\top}\mathbf{Q}^{\top}\mathbf{Q}\mathbf{L} = \mathbf{L}^{\top}\mathbf{L}.
\end{equation}
Since the message passing in \Cref{equation:MP} depends on $\mathbf{L}$ only through $\mathbf{L}^{\top}\mathbf{L}$, the node representations $\mathbf{h}_i^{(\ell)}$ are $O(3)$-invariant with respect to $\mathbf{L}$. 
By \Cref{equation:epsilonF}, the lattice denoising term is computed as $\hat{\epsilon}_{\mathbf{L}} = \mathbf{L} \varphi\left(\frac{1}{N}\sum_{i=1}^{N} \mathbf{h}_i^{(\mathcal{L})}\right)$. Under an orthogonal transformation $\mathbf{Q}$, we obtain:
\begin{align}
    \hat{\epsilon}_{\mathbf{L}}(\mathbf{QL}, \mathbf{F}, \mathbf{A}) 
    = (\mathbf{QL}) \varphi\left(\frac{1}{N}\sum_{i=1}^{N} \mathbf{h}_i^{(\mathcal{L})}\right) 
    &= \mathbf{Q} \cdot \mathbf{L} \varphi\left(\frac{1}{N}\sum_{i=1}^{N} 
    \mathbf{h}_i^{(\mathcal{L})}\right)
    \notag \\
    &= \mathbf{Q} \hat{\epsilon}_{\mathbf{L}}(\mathbf{L}, \mathbf{F}, \mathbf{A}),
\end{align}
which confirms the $O(3)$-equivariance of $\hat{\epsilon}_{\mathbf{L}}$.

\vspace{1mm}\textbf{Part II: $\hat{\epsilon}_{\mathbf{F}}$ is periodic translation invariant.} We show that the Fourier feature function $\psi: (-1,1)^3 \to [-1,1]^{3 \times K}$ is periodic translation invariant. Recall that the wrapping function $w(\cdot)$ maps each coordinate to $[0,1)$ via $w(x) = x - \lfloor x \rfloor$.

A key observation is that for any $a, b, t \in \mathbb{R}$, the difference of wrapped values satisfies
\begin{equation}
    w(a + t) - w(b + t) = (a - b) + z, \quad z \in \mathbb{Z},
\end{equation}
where $z = \lfloor b + t \rfloor - \lfloor a + t \rfloor$ is an integer that depends on the specific values of $a$, $b$, and $t$, but is always integral. This follows directly from the definition: $w(a+t) - w(b+t) = (a+t) - \lfloor a+t\rfloor - (b+t) + \lfloor b+t \rfloor = (a-b) + (\lfloor b+t\rfloor - \lfloor a+t\rfloor)$.

For odd $k$, let $\psi(x)[c,k] = \cos(2\pi m x_c)$ where $m = \lceil k/2 \rceil$. For any translation $\mathbf{t} \in \mathbb{R}^3$, applying the above identity component-wise yields:
\begin{align}
    \psi\bigl(w(f_j + \mathbf{t}) - w(f_i + \mathbf{t})\bigr)[c, k]
    &= \cos\Bigl(2\pi m \bigl(w(f_{j,c} + t_c) - w(f_{i,c} + t_c)\bigr)\Bigr) \notag \\
    = \cos\bigl(2\pi m (f_{j,c} - f_{i,c})\bigr) 
    &= \psi(f_j - f_i)[c, k],
\end{align}
where $z_c = \lfloor f_{i,c} + t_c \rfloor - \lfloor f_{j,c} + t_c \rfloor \in \mathbb{Z}$, and the third equality uses $\cos(\theta + 2\pi m z_c) = \cos(\theta)$ for any integer $m z_c$.
Similarly, for even $k$, let $\psi(x)[c,k] = \sin(2\pi m x_c)$. By the same identity and $\sin(\theta + 2\pi m z_c) = \sin(\theta)$:
\begin{align}
    \psi\bigl(w(f_j + \mathbf{t}) - w(f_i + \mathbf{t})\bigr)[c, k]
    &= \sin\Bigl(2\pi m \bigl(w(f_{j,c} + t_c) - w(f_{i,c} + t_c)\bigr)\Bigr) \notag \\
    = \sin\bigl(2\pi m (f_{j,c} - f_{i,c})\bigr)
    &= \psi(f_j - f_i)[c, k].
\end{align}
Since all inputs to the message passing layers in \Cref{equation:MP}---namely $\mathbf{L}^{\top}\mathbf{L}$ and $\psi(\sum_{j,j' \in \mathbf{e}}(f_{j'} - f_j))$---are periodic translation invariant, the node representations $\mathbf{h}_i^{(\ell)}$ inherit this invariance. Consequently, $\hat{\epsilon}_{\mathbf{F}}[:,i] = \varphi_{\mathbf{F}}(\mathbf{h}_i^{(\mathcal{L})})$ is periodic translation invariant.

\vspace{1mm}Combining Parts I and II, we conclude that the distribution generated by EH-Diff is invariant under the periodic $E(3)$ symmetry group.
\end{proof}

\subsection{Proof of \Cref{proposition:3.2}}
Before presenting the proof, we state a supporting lemma (proved in \cite{jiao2024crystal}) that will be used.
\begin{lemma}[\cite{jiao2024crystal}]\label{lemma:invariance}
Consider a generative Markov process defined as $p(x_0) = p(x_T) \int p(x_{0:T-1} \mid x_t) \, dx_{1:T}$. Suppose that the prior distribution $p(x_T)$ is $G$-invariant and that the Markov transitions $p(x_{t-1} \mid x_t)$, for $0 < t \leq T$, are $G$-equivariant. Then the marginal distribution $p(x_0)$ is $G$-invariant.
\end{lemma}

\begin{proof}
We prove the two claims of \Cref{proposition:3.2} separately.

\vspace{1mm}\textbf{Part I: $O(3)$-invariance of $p_\theta(\mathbf{L}_0)$.} The reverse transition is defined as:
\begin{equation}\label{eq:reverse_transition}
    p_\theta(\mathbf{L}_{t-1} \mid \mathbf{L}_t, \mathbf{F}_t, \mathbf{A}) = \mathcal{N}\left(\mathbf{L}_{t-1} \mid \boldsymbol{\mu}_t(\mathbf{L}_t), \sigma_t^2 \mathbf{I}\right),
\end{equation}
where the mean is given by:
\begin{equation}
    \boldsymbol{\mu}_t(\mathbf{L}_t) = a_t \left(\mathbf{L}_t - b_t \hat{\epsilon}_{\mathbf{L}}(\mathbf{L}_t, \mathbf{F}_t, \mathbf{A}, t)\right),
\end{equation}
with coefficients $a_t = \frac{1}{\sqrt{\alpha_t}}$, $b_t = \frac{\beta_t}{\sqrt{1-\bar{\alpha}_t}}$, $\sigma_t^2 = \beta_t \cdot \frac{1-\bar{\alpha}_{t-1}}{1-\bar{\alpha}_t}$, and $\bar{\alpha}_t = \prod_{s=1}^{t} \alpha_s$.
We verify the $O(3)$-equivariance of the transition. For any $\mathbf{Q} \in O(3)$, by the $O(3)$-equivariance of $\hat{\epsilon}_{\mathbf{L}}$ established in \Cref{proposition:3.1}:
\begin{align}
    \boldsymbol{\mu}_t(\mathbf{Q}\mathbf{L}_t) 
    &= a_t \left(\mathbf{Q}\mathbf{L}_t - b_t \hat{\epsilon}_{\mathbf{L}}(\mathbf{Q}\mathbf{L}_t, \mathbf{F}_t, \mathbf{A}, t)\right) \notag \\
    &= a_t \left(\mathbf{Q}\mathbf{L}_t - b_t \mathbf{Q}\hat{\epsilon}_{\mathbf{L}}(\mathbf{L}_t, \mathbf{F}_t, \mathbf{A}, t)\right) \notag \\
    &= \mathbf{Q} \cdot a_t \left(\mathbf{L}_t - b_t \hat{\epsilon}_{\mathbf{L}}(\mathbf{L}_t, \mathbf{F}_t, \mathbf{A}, t)\right) 
    = \mathbf{Q} \boldsymbol{\mu}_t(\mathbf{L}_t).
\end{align}
For a Gaussian random variable $\mathbf{L} \sim \mathcal{N}(\boldsymbol{\mu}, \sigma^2 \mathbf{I})$ and orthogonal $\mathbf{Q}$, we have:
\begin{equation}
    \mathbf{Q}\mathbf{L} \sim \mathcal{N}\left(\mathbf{Q}\boldsymbol{\mu}, \mathbf{Q}(\sigma^2 \mathbf{I})\mathbf{Q}^{\top}\right) = \mathcal{N}\left(\mathbf{Q}\boldsymbol{\mu}, \sigma^2 \mathbf{I}\right),
\end{equation}
where the last equality uses $\mathbf{Q}\mathbf{Q}^{\top} = \mathbf{I}$.
Combining these results, the transition probability satisfies:
\begin{align}
    p_\theta(\mathbf{Q}\mathbf{L}_{t-1} \mid \mathbf{Q}\mathbf{L}_t, \mathbf{F}_t, \mathbf{A}) 
    &= \mathcal{N}\left(\mathbf{Q}\mathbf{L}_{t-1} \mid \boldsymbol{\mu}_t(\mathbf{Q}\mathbf{L}_t), \sigma_t^2 \mathbf{I}\right) \notag \\
    &= \mathcal{N}\left(\mathbf{Q}\mathbf{L}_{t-1} \mid \mathbf{Q}\boldsymbol{\mu}_t(\mathbf{L}_t), \sigma_t^2 \mathbf{I}\right) \notag \\
    &= \mathcal{N}\left(\mathbf{L}_{t-1} \mid \boldsymbol{\mu}_t(\mathbf{L}_t), \sigma_t^2 \mathbf{I}\right) \notag \\
    &= p_\theta(\mathbf{L}_{t-1} \mid \mathbf{L}_t, \mathbf{F}_t, \mathbf{A}),
\end{align}
where the third equality follows from the change of variables $\mathbf{L}_{t-1} = \mathbf{Q}^{\top}(\mathbf{Q}\mathbf{L}_{t-1})$ with $|\det(\mathbf{Q})| = 1$ (since $\mathbf{Q}$ is orthogonal, the Jacobian determinant is unity), together with the isotropic covariance $\sigma_t^2 \mathbf{I}$ being invariant under orthogonal transformations. This confirms the $O(3)$-equivariance of the reverse transition.
Since the prior $p(\mathbf{L}_T) = \mathcal{N}(\mathbf{0}, \mathbf{I})$ is $O(3)$-invariant and the transitions are $O(3)$-equivariant, \Cref{lemma:invariance} implies that $p_\theta(\mathbf{L}_0)$ is $O(3)$-invariant.

\vspace{1mm}\textbf{Part II: Periodic translation invariance of $p_\theta(\mathbf{F}_0)$.} For fractional coordinates, the reverse process employs the score-based predictor-corrector method with wrapped normal distributions. The score function is given by $\nabla_{\mathbf{F}_t} \log q(\mathbf{F}_t \mid \mathbf{F}_0)$, which is approximated by $\hat{\epsilon}_{\mathbf{F}}$.
For any translation $\mathbf{t} \in \mathbb{R}^{3 \times 1}$, define $\tilde{\mathbf{F}}_t = w(\mathbf{F}_t + \mathbf{t}\mathbf{1}^{\top})$ where $w(\cdot)$ is the element-wise wrapping function. By the periodic translation invariance of $\hat{\epsilon}_{\mathbf{F}}$ from \Cref{proposition:3.1}:
\begin{equation}
    \hat{\epsilon}_{\mathbf{F}}(\tilde{\mathbf{F}}_t, \mathbf{L}_t, \mathbf{A}, t) = \hat{\epsilon}_{\mathbf{F}}(\mathbf{F}_t, \mathbf{L}_t, \mathbf{A}, t).
\end{equation}
The prior distribution $q(\mathbf{F}_T) = \mathcal{U}([0,1)^{3 \times N})$ is uniform and hence periodic translation invariant. Since both the prior and the score function are periodic translation invariant, the Langevin dynamics in the reverse process preserves this invariance at each step. Therefore, $p_\theta(\mathbf{F}_0)$ is periodic translation invariant.
The learned marginal distributions $p_\theta(\mathbf{L}_0)$ and $p_\theta(\mathbf{F}_0)$ are $O(3)$-invariant and periodic translation invariant, respectively.
\end{proof}

\section{Related Works}

\paragraph{Crystal Structure Prediction} Traditional CSP methods combine DFT \cite{kohn1965self} with optimization algorithms \cite{pickard2011ab, yamashita2018crystal} to locate energy minima, but at substantial computational cost. Machine learning surrogates \cite{cheng2022crystal} improve scalability, and more recently, deep generative models using representations such as voxel grids \cite{gebauer2022inverse}, distance matrices \cite{hu2020distance}, and direct 3D coordinates \cite{nouira2018crystalgan} have emerged as efficient alternatives for CSP \cite{jiao2024crystal, song2024diffusion}.

\paragraph{Equivariant Models} $E(3)$-equivariant GNNs \cite{schutt2018schnet, satorras2021n} have become standard for modeling physical systems. For periodic materials, \citet{xie2018crystal} introduced multi-graph edge construction across lattice cells, and subsequent work \cite{yan2022periodic, jiao2024crystal} incorporated periodic encoding into Transformer architectures.


\end{document}